\documentclass[11pt]{article}
\usepackage{amssymb}
\usepackage{amsmath}
\usepackage{comment}
\usepackage{amsthm}
\usepackage{graphicx}
\usepackage{subfigure}
\usepackage{makeidx}
\usepackage{multicol,changepage}
\usepackage{url}
\usepackage[toc,page]{appendix}
\usepackage{xcolor}
\definecolor{fgcolor}{rgb}{0.345, 0.345, 0.345}
\usepackage{enumitem} 
\usepackage{tikz}
\usepackage{alltt}
\usepackage{algorithm2e}
\usepackage{geometry}
\usepackage{framed, float}

\makeatletter
 {\par\unskip\endMakeFramed%
 \at@end@of@kframe}

\tikzset{
    place/.style={
        circle,
        thick,
        draw=black,
        fill=gray!50,
        minimum size=6mm,
    },
        state/.style={
        circle,
        thick,
        draw=blue!75,
        fill=blue!20,
        minimum size=6mm,
    },
}

\newtheorem{theorem}{theorem}%[section]

\newtheorem{conjecture}[theorem]{Conjecture}

\newtheorem{proposition}[theorem]{Proposition}

\newtheorem{definition}[theorem]{Definition}

\newtheorem{remark}[theorem]{Remark}

\newtheorem{example}[theorem]{Example}

\newcommand{\Z}{\mathbb Z}

\begin{document}
\title{Connecting Tables with Allowing Negative Cell Counts}

\author{Ruriko Yoshida \and David Barnhill}

\date{}

\maketitle

\abstract{It is well-known that computing a Markov basis for a discrete loglinear model is very hard in general. Thus, we focus on connecting tables in a fiber via a subset of a Markov basis and in this paper, we consider connecting tables if we allow cell counts in each tale to be $-1$.  In this paper we show that if a subset of a Markov basis connects all tables in the fiber which contains a table with all ones, then moves in this subset connect tables in the fiber if we allow cell counts to be $-1$. In addition, we show that in some cases under the no-three-way interaction model,  we can connect tables by all basic moves of $2 \times 2 \times 2$ minors with allowing $X_{ijk} \geq -1$.    We then apply this Markov Chain Monte Carlo (MCMC) scheme to an empirical data on Naval officer and enlisted population. Our computational experiments show it works well and  we end with the conjecture on the no-three-way interaction model.}

\section{Introduction}

Conducting the goodness of fit test for discrete exponential family models with sparse contingency tables is a challenging task since we cannot use asymptotic methods, especially when we have large tables.  To do so, one applies sampling procedures, such as Markov Chain Monte Carlo (MCMC) procedures, from the state space \cite{Besag:1989uq,chen2005,Diaconis:1985kx,Guo:1992fk}. 

Even though Besag and Clifford argued that irreducibility of a Markov chain is not essential in \cite{Besag:1989uq}, that view is not conventional since it is very difficult to understand a chain over sparse tables in high dimensions.
To understand the irreducibility of a Markov chain over the set of all possible contingency tables with given constraints, Diaconis and Sturmfels introduced a notion of a {\em Markov basis} \cite{DS}.  Diaconis and Sturmfels showed that a Markov basis, a generating set of a toric ideal associated with the design matrix for a discrete exponential family model, is the set of moves which guarantee to connect all tables for the given model \cite{DS}.  The most significant benefit in computing a Markov basis for a discrete exponential family model is that if we have a Markov basis for the model, we do not have to re-compute it since moves in a Markov basis connect all tables for any observed tables under the model.  However, it is computationally expensive to compute a Markov basis in terms of the dimension of a table.  In fact, \cite{deloera05} showed that the number of elements in a Markov basis for the no-three-way interaction model can be arbitrarily large.   This means that it is not feasible to compute a Markov basis for large tables under the no-three-way interaction model. %Therefore, in this research, we focus on the connectivity of a Markov chain for $I \times J \times K$ contingency tables under the no-three-way interaction model. 
Then, Hara et al.~conducted study of employing MCMC with a {\em lattice basis} which is much easier to compute than a Markov basis \cite{Hara2011RunningMC}.  However, with a lattice basis, we cannot guarantee that a chain is connected so that the estimated distribution of statistics might be biased.  

To try to circumvent the difficulty of computing a
Markov basis which may be large,
\cite{chen-dinwoodie-dobra-huber2005,Bunea:2000fk} 
studied computing a smaller set of moves, such as a lattice basis, by allowing cell counts of the
contingency table to be negative one ($-1$).  Motivated by their work, in this paper, we consider cases when we can connect a Markov chain if we allow cell counts to be $-1$ including some cases on the {\em no-three-way interaction model} with the {\em basic moves}.

A basic move of a $2 \times 2\times 2$ minor for the no-three-way interaction model is a $I \times J \times K$ contingency table with $\{0, \pm 1\}$ entries such that for distinct $i, i' \in \{1, \ldots , I\}$, distinct $j, j' \in \{1, \ldots , J\}$, and distinct $k, k' \in \{1, \ldots , K\}$,  
\begin{eqnarray*}
i & \begin{array}{|c||c|c|}\hline
 & j & j'\\\hline \hline
k & 1 & -1\\\hline
k' & -1 & 1\\\hline
\end{array}\\
i' & \begin{array}{|c||c|c|}\hline
 & j & j'\\\hline \hline
k & -1 & 1\\\hline
k' & 1 & -1\\\hline
\end{array}\\
\end{eqnarray*}
and all zeros otherwise.  %One can show that if $I = 2$, $J \geq 2$, and $K \geq 2$, then basic moves of all $2 \times 2 \times 2$ minors connect all tables if we allow cell counts to be negative one. 
In 2018, Lee showed that if $I = 3$, $J = 3$, and $K \geq 2$, then 
%Then, in this paper, we show that 
basic moves connect all tables if we allow cell counts $X_{ijk}$ to be $-1$ when we run a Markov chain. For $I > 3$, $J > 3$, and $K \geq 2$, it is still an open problem that basic moves connect all tables if we allow cell counts $X_{ijk}$ to be $-1$ when we run a Markov chain. 

\begin{comment}
The trade off to this approach is longer
running time of the Markov chains.  
Even using a standard MCMC approach, to sample a table independently from the distribution,
Markov chains can take a long time to 
converge to a stationary distribution in order to satisfy the
independent assumption. 
Since we allow cell counts to drop to $-1$, when we run a Markov chain, it is possible to slow down a mixing time of a Markov chain.  
In order to speed up mixing time we propose to combine a {\em simulated annealing (SA)} approach \cite{doi:10.1126/science.220.4598.671} with a MCMC with basic moves allowing cell counts to be negative one. With SA, we can push a chain more inside of the state space, called a {\em fiber}, of contingency tables with fixed margins so that we can avoid having a fat tails of a distribution of $\chi^2$ statistics.  
\end{comment}

The contribution of this paper is the following:
\begin{itemize}
    \item We show that under a fixed model, if a set of moves $M$ connects all tables in a sample space, the set of all tables with fixed margins, which contains a table with all cell counts equal to $1$, then $M$ connects all tables in a sample space with any fixed margins if we allow cell counts to be $-1$. 
    \item We show that under the no-three-way interaction model, the set of basic moves of $2\times 2 \times 2$ minors connects tables in a sample space with a fixed margins if we allow cell counts to be $-1$ in the case of $I = 2$.  Then we also show that the set of basic moves connects all tables in the sample space with any fixed margins if we allow cell counts to be $-1$ for $I = 3$ and $J = 4$.  
    \item We apply our result to the Navy personnel dataset obtained from U.S. NAVY DEMOGRAPHIC DATA, 2022.  
\end{itemize} 

This paper is organized as follows:  In Section \ref{sec:basic}, we set up basics for our results in this paper.  Then, in Section \ref{sec:thm}, we show some theoretical conditions allowing us to connect all tables in the given fiber with a subset of a Markov basis if we allow cell counts to be $-1$.  Additionally, we show cases of the three way contingency tables under the no-three-way interaction model that we can connect tables in the fiber by a set of all basic moves while allowing cell counts to be $-1$.
%In Section \ref{sec:hybrid}, we propose a hybrid scheme between MCMC with basic moves allowing cell counts to be negative one with SA and HAR sampling. 
Then we apply our hybrid scheme, that is, a Markov chain with a set of basic moves by allowing cell counts to be negative one combined with SA to an empirical data on U.S. Navy (USN) personnel in Section \ref{sec:comp}.  We end with discussion and open problems.  

\section{Preliminaries}\label{sec:basic}
In this section we set up some basic notation and definitions for our results in this paper.  

Let $A \in \Z^{m\times d}$ and $b \in \Z^m$.  We define $b-$fiber of
$A$ is
\[
\mathcal{F}_{A, b}:= \{ u\in \Z_{\geq 0}^d: A u = b\}.  
\]
For a finite
$\mathcal{M} \subset \ker A \cap \Z^n$, we define the fiber
graph $\mathcal{G}(\mathcal{F}_{A,b}, \mathcal{M})$ of $A$ as graph having the
$b$-fiber as its vertex set and the set 
\[
\{(u, v): u, v \in \mathcal{F}_{A, b}, \, (u - v) \in \pm \mathcal{M}\} 
\]
as its edge set.

\begin{definition}
We fix a matrix $A \in \Z^{m\times d}$.  We call $\mathcal{M} \subset \ker A \cap \Z^n$ as a {\em Markov basis} for a matrix $A$ if the fiber
graph $\mathcal{G}(\mathcal{F}_{A,b}, \mathcal{M})$ of $A$ is connected for any $b \in \Z^m$.
\end{definition}
%Sturmfels and Diaconis showed how to compute a Markov basis for a matrix $A$ using tools from Algebraic Geometry \cite{DS}.

In this paper, we focus on assessing interaction between three discrete random variables
$X, \, Y, \, Z$ with finite levels, where $X \in \{1, \ldots , I\}$,
$Y \in \{1, \ldots , J\}$, and $Z \in \{1, \ldots , K\}$ for $I, \,
J, \, K \in \mathbb{N}$.  In order to conduct statistical analysis on such
interactions,  a commonly used model is called a {\em log linear
  model}.  A log linear model is a generalized linear model and it
forms a class of discrete exponential family.  

\begin{definition}
Suppose $A \in \Z^{m \times d}$ be a non-negative integral matrix such
that 
\[
\sum_{i = 1}^m a_{i1} = \ldots = \sum_{i = 1}^m a_{id}
\]
where $a_{ij}$ is the $(i, j)$th element of the matrix $A$.  
If its likelihood function $L(\theta| \bf{n})$ for an observation $\bf{n}$ can be written as
\[
L(\theta| \mathbf{n}) = c \cdot \theta ^{\;A\bf{n}}
\]

\noindent where $c$ is a normalized constant, then the model is called a {\em
  log linear model} and we call such matrix $A$ as the {\em design matrix} for a model.
\end{definition}

Applying the notation of the log-linear model notation to the
three-way contingency table, we have the parameterization $n_{ijk}
\sim Pois(\mu_{ijk})$, the Poisson distribution with the parameter $\mu_{ijk}$, such that:
\[
\log(\mu_{ijk})=\lambda+\lambda_i^X+\lambda_j^Y+\lambda_k^Z+\lambda_{ij}^{XY}+\lambda_{ik}^{XZ}+\lambda_{jk}^{YZ}+\lambda_{ijk}^{XYZ}.
\]
If we assume there is no interaction between three categorical
variables, namely, 
\[
\lambda_{ijk}^{XYZ} = 0,
\]
for all $i = 1, \ldots , I$, $j = 1, \ldots , J$, and $k = 1, \ldots ,
K$, then we call this model 
the {\em no-three-way interaction model}.

It is easy to show that under the no-three-way interaction model, we have 
the $b-$fiber of $A$ with a given observed table ${\bf n}$ is 
\[
\mathcal{F}_{A, b}:= \left\{{\bf u} \Big| \sum_{i = 1}^m u_{ijk} = \sum_{i = 1}^m n_{ijk}, \, \sum_{j = 1}^n
u_{ijk} = \sum_{j = 1}^n n_{ijk}, \, \sum_{k = 1}^l u_{ijk} = \sum_{k
  = 1}^l n_{ijk} \right\} =: \{u \in \Z^{IJK}_+| Au = b\},
\]
where $n_{ijk}$ is the cell count for the $(i, j, k)$th cell in ${\bf n}$.
\begin{example}
Suppose we have a $2 \times 2 \times 2$ table such that
\[
\begin{array}{|c|c|}\hline
     X_{111}&X_{112}  \\\hline
     X_{121}&X_{122}  \\\hline
\end{array}
\ \
\begin{array}{|c|c|}\hline
     X_{211}&X_{212}  \\\hline
     X_{221}&X_{222}  \\\hline
\end{array} .
\]
The the design matrix under the no-three-way interaction model is 
\[
A = \left(
\begin{array}{cccccccc}
     1 & 1 & 0& 0 &  0 & 0 & 0& 0  \\
     0 & 0 & 1 & 1 & 0& 0 &  0 & 0 \\
     0 & 0 & 0& 0 &  1 & 1 & 0& 0  \\
     0 & 0 & 0& 0 &  0 & 0 & 1 & 1 \\
     1 & 0 & 1 & 0 & 0 & 0 & 0 & 0\\
     0 & 1 & 0 & 1 & 0 & 0 & 0 & 0\\
     0 & 0 & 0 & 0 & 1 & 0 & 1 & 0\\
     0 & 0 & 0 & 0 & 0 & 1 & 0 & 1\\
     1 & 0 & 0 & 0 & 1 & 0 & 0 & 0\\
     0 & 1 & 0 & 0 & 0& 1 &  0 & 0\\
     0 & 0 & 1 & 0 & 0 & 0 & 1 & 0\\
     0 & 0 & 0 & 1 & 0 & 0 & 0 & 1 \\
\end{array}
\right) .
\]
Then, suppose we have a $2 \times 2 \times 2$ table such that
\[
\begin{array}{|c|c|}\hline
     1&1 \\\hline
     1&1  \\\hline
\end{array}
\ \
\begin{array}{|c|c|}\hline
     1&1 \\\hline
     1&1 \\\hline
\end{array} .
\]
Then the vector of margins is
\[
b = (2, 2, 2, 2, 2, 2, 2, 2, 2, 2, 2, 2)^T ,
\]
where $x^T$ is the transpose of a vector $x \in \mathbb{R}^m$.
\end{example}

Now we define a \em{basic move}.
\begin{definition}
Let $m$ be a $I \times J \times K$ table such that  
\begin{eqnarray*}
i & \begin{array}{|c||c|c|}\hline
 & j & j'\\\hline \hline
k & 1 & -1\\\hline
k' & -1 & 1\\\hline
\end{array}\\
i' & \begin{array}{|c||c|c|}\hline
 & j & j'\\\hline \hline
k & -1 & 1\\\hline
k' & 1 & -1\\\hline
\end{array}\\
\end{eqnarray*}
where $1 \leq i, i' \leq I, \, 1 \leq j, j' \leq
J, \, 1 \leq k, k' \leq K$, $i \not = i'$, $j \not = j'$, and $k \not
= k'$, and other cells are all zero.  We call $m$ a {\em basic move}. 
We denote this table $m$ as $(i, i'; j, j'; k, k')$.
\end{definition}

\section{Connectivity of a Markov Chain via Subset of Markov Basis}\label{sec:thm}

In this section, we show the main result for the connectivity of a Markov chain if we allow cell counts to drop to $-1$. 

\begin{theorem}\label{th:greater}
Suppose $A$ is a $0-1$ matrix and suppose we have a set of moves $M$ such that moves in $M$ connect all points in the fiber $\mathcal{F}_{A, A{\mathbb{I}}}$, where $\mathbb{I} = (1, \ldots , 1)$ the vector with all ones. 
Suppose $\mathcal{F}_{A, b}$ contains a point $x^0:=(x^0_1, \ldots , x^0_n)$ such that $x^0_i \geq 1$ for $i = 1, \ldots , n$.  Then moves in $M$ connect all points in $\mathcal{F}_{A, b}$.
\end{theorem}
\begin{proof}
Let 
\[
I_M = \langle x^{m_+} - x^{m_-}: m \in M\rangle .
\]
Let $b':= A \mathbb{I}$.  Let $I_m := \langle x_k\rangle_{A_m, n}$ be the monomial ideal generated by all the indeterminates for the cells that contribute to a margin $m$ and let $I^{b'}_M = \langle x_{i_1}x_{i_2}\ldots x_{i_{b'}}\rangle_{A_{m}, i_k}$.
Suppose $u, v \in \mathcal{F}_{A, b'}$.  Then since $A$ is a $0-1$ matrix, each monomial $x^u$ and $x^v$ are in $I^{b'}_m$.  Since all moves in $M$ connect all points in $\mathcal{F}_{A, b'=A{\mathbb{I}}}$,
\[
\mathcal{I}_A \cap \bigcap_m I^{b'}_m \subset I_M.
\]
Then we apply Proposition 0.4.1 in \cite{CDY}.  
\end{proof}

Then we have our main theorem.
\begin{theorem}\label{th:main}
Suppose $A$ is a $0-1$ matrix and suppose we have a set of moves $M$ such that moves in $M$ connect all points in the fiber $\mathcal{F}_{A, A{\mathbb{I}}}$, where $\mathbb{I} = (1, \ldots , 1)$ the vector with all ones.
Consider a contingency table
$X \in \Z_{\geq 0}^n$.  We assume that $ (b - A
\mathbb{I}_n) \geq 0$, where $A$ is the design matrix for a given model for contingency tables, and $b$ is a vector of the margins.
If we allow $X_{i} \geq -1$, then all moves in $M$ connect all
tables in ${\mathcal{F}}_{A, b}$.
\end{theorem}
\begin{proof}
We consider the fiber 
\[
\begin{array}{ccc}
\bar{\mathcal{F}}_{A, b}&:= &\{ u\in \Z^n: A u = b, \, u \geq -\mathbb{I}_n\}\\
&=& \{ u + \mathbb{I}_n\in \Z_{\geq 0}^n: A (u +
\mathbb{I}_n)= b\}\\
&= & \{ u\in \Z_{\geq 0}^n: A u = (b - A
\mathbb{I}_n)\}.\\
\end{array}
\]
Since we have $ (b - A
\mathbb{I}_n) \geq 0$,  this means that $b \geq A\mathbb{I}_n$,
Then we apply Theorem \ref{th:greater}. 
\end{proof}
We apply Theorem \ref{th:main} to the no-three-way interaction model.
\begin{proposition}\label{pro:main}
    Consider the design matrix for three-way contingency tables under the no-three-way interaction model.  suppose we have a set of moves $M$ such that moves in $M$ connect all tables in the fiber which contains the table with all cell counts equal to one.  Then, if we allow $X_{ijk} \geq -1$ for all $i = 1, \ldots , I$, $j = 1, \ldots , J$, and $k = 1, \ldots , K$, then M connects all tables in the fiber. 
\end{proposition}
\begin{proof}
    We simply apply Theorem \ref{th:main}.
\end{proof}

Note that we cannot apply Proposition \ref{pro:main} to the set of all basic moves since the set of basic moves does not connect tables in the fiber which contains a table with all ones.  As an example, we consider the case of $I = J = K = 3$.  
Suppose we have a $3 \times 3\times 3$ table such that
\[
\begin{array}{|c|c|c|}\hline
     3&0 & 0 \\\hline
     0&3 & 0  \\\hline
     0 & 0 & 3\\\hline
\end{array}
\ \
\begin{array}{|c|c|c|}\hline
     0&3 & 0 \\\hline
     0&0 & 3  \\\hline
     3 & 0 & 0\\\hline
\end{array} 
\ \
\begin{array}{|c|c|c|}\hline
     0&0 & 3 \\\hline
     3&0 & 0  \\\hline
     0 & 3 & 0\\\hline
\end{array} .
\]
Then note that the set of basic moves does not connect this table to any other tables with the same margins.  However, we can show that if we allow cell counts to be $-1$, then we can connect this table to other tables in the fiber $\mathcal{F}_{A, b}$  \cite{10945-59705}.
Therefore, in the next section, we consider some cases which we can connect three-way contingency tables under the no-three-way interaction model with the set of basic moves by allowing cell counts to be $-1$.

\section{Connecting $I \times J\times K$ Tables with Basic Moves}

In this section, we consider three-way $I \times J\times K$ tables under the no-three-way interaction model with the set of basic moves. 

%\begin{remark}
Note that, by Lemma 12.2 in \cite{Sturmfels1995GrobnerBA}, for some positive integer $t$, if we allow $X_{ijk} \geq -t$, then all moves in $M$ connect all
tables in $\mathcal{F}_{A, b}$ for a three way $I\times J\times K$ contingency table
$X_{ijk}$ under the no-three way interaction model.
%\end{remark}

\begin{remark}
Suppose ${\bf b}={\bf b}_+ - {\bf b}_-$ is in a Markov basis for a matrix $A$. If we can show that a set of basic moves connect from ${\bf b}_+$ and ${\bf b}_-$ for any move ${\bf b}={\bf b}_+ - {\bf b}_-$ in a Markov basis by allowing cell counts to be negative one, then basic moves connect all tables in $\mathcal{F}_{A, b}$ for any $b \in \mathbb{Z}^m_{\geq 0}$ by allowing cell counts to be $-1$.
\end{remark}

\begin{example}
First we consider a $3 \times 4 \times 6$ table and we consider an
``indispensable move'', which is a necessary move in all Markov bases for a fixed matrix $A$, such that
\begin{eqnarray*}
& \begin{array}{|c|c|c|c|c|c|}\hline
+1 & -1 & 0 & 0 & 0 & 0\\\hline 
0 & +1 & -1 & 0 & 0 & 0\\\hline
0 & 0 & +1 & 0 & 0 & -1\\\hline
-1 & 0 & 0 & 0 & 0 & +1\\\hline
\end{array} \, \, 
\begin{array}{|c|c|c|c|c|c|}\hline
0 & +1 & 0 & -1 & 0 & 0\\\hline 
0 & -1 & 0 & 0 & +1 & 0\\\hline
0 & 0 & 0 & +1 & 0 & -1\\\hline
0 & 0 & 0 & 0 & -1 & +1\\\hline
\end{array}\, \, 
\begin{array}{|c|c|c|c|c|c|}\hline
-1 & 0 & 0 & +1 & 0 & 0\\\hline 
0 & 0 & +1 & 0 & -1 & 0\\\hline
0 & 0 & -1 & -1 & 0 & +2\\\hline
+1 & 0 & 0 & 0 & +1 & -2\\\hline
\end{array}\\
= & 
\begin{array}{|c|c|c|c|c|c|}\hline
1 & 0 & 0 & 0 & 0 & 0\\\hline 
0 & 1 & 0 & 0 & 0 & 0\\\hline
0 & 0 & 1 & 0 & 0 & 0\\\hline
0 & 0 & 0 & 0 & 0 & 1\\\hline
\end{array} \, \, 
\begin{array}{|c|c|c|c|c|c|}\hline
0 & 1 & 0 & 0 & 0 & 0\\\hline 
0 & 0 & 0 & 0 & 1 & 0\\\hline
0 & 0 & 0 & 1 & 0 & 0\\\hline
0 & 0 & 0 & 0 & 0 & 1\\\hline
\end{array}\, \, 
\begin{array}{|c|c|c|c|c|c|}\hline
0 & 0 & 0 & 1 & 0 & 0\\\hline 
0 & 0 & 1 & 0 & 0 & 0\\\hline
0 & 0 & 0 & 0 & 0 & 2\\\hline
1 & 0 & 0 & 0 & 1 & 0\\\hline
\end{array}\\
- & 
\begin{array}{|c|c|c|c|c|c|}\hline
0 & 1 & 0 & 0 & 0 & 0\\\hline 
0 & 0 & 1 & 0 & 0 & 0\\\hline
0 & 0 & 0 & 0 & 0 & 1\\\hline
1 & 0 & 0 & 0 & 0 & 0\\\hline
\end{array} \, \, 
\begin{array}{|c|c|c|c|c|c|}\hline
0 & 0 & 0 & 1 & 0 & 0\\\hline 
0 & 1 & 0 & 0 & 0 & 0\\\hline
0 & 0 & 0 & 0 & 0 & 1\\\hline
0 & 0 & 0 & 0 & 1 & 0\\\hline
\end{array}\, \, 
\begin{array}{|c|c|c|c|c|c|}\hline
1 & 0 & 0 & 0 & 0 & 0\\\hline 
0 & 0 & 0 & 0 & 1 & 0\\\hline
0 & 0 & 1 & 1 & 0 & 0\\\hline
0 & 0 & 0 & 0 & 0 & 2\\\hline
\end{array}\\
=: & {\bf b^1_{+}} - {\bf b^1_{-}}.
\end{eqnarray*}

If we allow $X_{ijk} \geq -1$ for $1 \leq i \leq I, 1 \leq j \leq J, 1\leq
k \leq K$, then we can show that 
\[
\begin{array}{cc}
   {\bf b^1_{-}} =   &  {\bf b^1_{+}} - (2,3;3, 4; 5,6) -(1,3;3,4;1,6) +
 (2,3;2,3;3,5)\\
     & - (2,3;1,3;1,4) - (1,2;2,3;2,3) - (1,2;1,3;1,2).
\end{array}
% {\bf b^1_{-}} =  {\bf b^1_{+}} - (2,3;3, 4; 5,6) -(1,3;3,4;1,6) +
% (2,3;2,3;3,5) - (2,3;1,3;1,4) - (1,2;2,3;2,3) - (1,2;1,3;1,2).
\]

In addition we consider another indispensable move such that
\begin{eqnarray*}
& \begin{array}{|c|c|c|c|c|c|}\hline
+1 & -1 & 0 & 0 & 0 & 0\\\hline 
0 & +1 & -1 & 0 & 0 & 0\\\hline
0 & 0 & +1 & -1 & 0 & 0\\\hline
-1 & 0 & 0 & +1 & 0 & 0\\\hline
\end{array} \, \, 
\begin{array}{|c|c|c|c|c|c|}\hline
-1 & 0 & 0 & 0 & 0 & +1\\\hline 
0 & 0 & +1 & 0 & -1 & 0\\\hline
0 & 0 & -1 & 0 & 0 & +1\\\hline
+1 & 0 & 0 & 0 & +1 & -2\\\hline
\end{array}\, \, 
\begin{array}{|c|c|c|c|c|c|}\hline
0 & +1 & 0 & 0 & 0 & -1\\\hline 
0 & -1 & 0 & 0 & +1 & 0\\\hline
0 & 0 & 0 & +1 & 0 & -1\\\hline
0 & 0 & 0 & -1 & -1 & +2\\\hline
\end{array}\\
= & 
\begin{array}{|c|c|c|c|c|c|}\hline
1 & 0 & 0 & 0 & 0 & 0\\\hline 
0 & 1 & 0 & 0 & 0 & 0\\\hline
0 & 0 & 1 & 0 & 0 & 0\\\hline
0 & 0 & 0 & 1 & 0 & 0\\\hline
\end{array} \, \, 
\begin{array}{|c|c|c|c|c|c|}\hline
0 & 0 & 0 & 0 & 0 & 1\\\hline 
0 & 0 & 1 & 0 & 0 & 0\\\hline
0 & 0 & 0 & 0 & 0 & 1\\\hline
1 & 0 & 0 & 0 & 1 & 0\\\hline
\end{array}\, \, 
\begin{array}{|c|c|c|c|c|c|}\hline
0 & 1 & 0 & 0 & 0 & 0\\\hline 
0 & 0 & 0 & 0 & 1 & 0\\\hline
0 & 0 & 0 & 1 & 0 & 0\\\hline
0 & 0 & 0 & 0 & 0 & 2\\\hline
\end{array}\\
- & 
\begin{array}{|c|c|c|c|c|c|}\hline
0 & 1 & 0 & 0 & 0 & 0\\\hline 
0 & 0 & 1 & 0 & 0 & 0\\\hline
0 & 0 & 0 & 1 & 0 & 0\\\hline
1 & 0 & 0 & 0 & 0 & 0\\\hline
\end{array} \, \, 
\begin{array}{|c|c|c|c|c|c|}\hline
1 & 0 & 0 & 0 & 0 & 0\\\hline 
0 & 0 & 0 & 0 & 1 & 0\\\hline
0 & 0 & 1 & 0 & 0 & 0\\\hline
0 & 0 & 0 & 0 & 0 & 2\\\hline
\end{array}\, \, 
\begin{array}{|c|c|c|c|c|c|}\hline
0 & 0 & 0 & 0 & 0 & 1\\\hline 
0 & 1 & 0 & 0 & 0 & 0\\\hline
0 & 0 & 0 & 0 & 0 & 1\\\hline
0 & 0 & 0 & 1 & 1 & 0\\\hline
\end{array}\\
=: & {\bf b^2_{+}} - {\bf b^2_{-}}.
\end{eqnarray*}

If we allow $X_{ijk} \geq -1$ for $1 \leq i \leq I, 1 \leq j \leq J, 1\leq
k \leq K$, then we can show that 
\[
\begin{array}{cc}
  {\bf b^2_{-}} =     &  {\bf b^2_{+}} + (2,3;2, 4; 5,6) +(2,3;3,4;4,6) +
 (2,3;1,2;2,6)\\
     & - (1,2;3,4;1,4) - (1,2;1,3;1,3) + (1,2;1,2;2,3).
\end{array}
% {\bf b^2_{-}} =  {\bf b^2_{+}} + (2,3;2, 4; 5,6) +(2,3;3,4;4,6) +
% (2,3;1,2;2,6) - (1,2;3,4;1,4) - (1,2;1,3;1,3) + (1,2;1,2;2,3).
\]

\end{example}

\begin{remark}\label{rm:indis}
One notices that by the way we use basic moves to connect from ${\bf
  b_+^1}$ to ${\bf b_-^1}$ and ${\bf b_+^2}$ to ${\bf b_-^2}$ we
cannot connect by allowing just one cell to be $-1$.  We have to allow at
least three cells to be $-1$.  So it is interesting to find out
the upper bounds $K$ such that all moves in $M$ connect all
tables in $F$ if  we allow at most $K$ cells $X_{ijk} \geq -1$.
\end{remark}

Dinwoodie et al.~discussed connecting tables via allowing cell counts to be negative one and they discussed a condition for a set of moves so that it connects tables in the fiber by allowing cell counts to be negative one (Proposition 0.2.1 \cite{Ian2008}).  However, for the no-three-way interaction model, the {\em ideal} generated by binomials with exponents for the positive part and negative part of all possible basic moves for $3 \times 3 \times 4$ is not radical.  Therefore, we cannot use Proposition 0.2.1 for the connectivity of tables with allowing cell counts to be $-1$. However, using a known Markov basis for $3 \times 3 \times K$, for $K \geq 2$, by Aoki and Takemura in \cite{aoki-takemura-2003anz}, we can show that basic moves of all $2 \times 2 \times 2$ minors can connect tables for $3 \times 3 \times 4$ case.

We can also show that a set of basic moves $M$ connects all tables in the fiber ${\mathcal{F}}_{A, b}$ for $A$ a design matrix for the model and $b$ a vector of the margins if we allow a cell count $X_{ijk} \geq -1$ under the no-three-way interaction model.
To prove this, we use elements of all indispensable moves for  $3 \times 4\times K$ tables in \cite{Aoki03thelist}.  Then as we discussed in Remark \ref{rm:indis}, we can show the positive component and negative component of each indispensable move can be connected via basic moves if we allow cell counts to be $-1$.
\begin{proposition}
Consider a $3 \times 4\times K$ table $X$ for $K = 2, \ldots $ under the no-three-way interaction model.  Then a set of basic moves $M$ connects all tables in the fiber ${\mathcal{F}}_{A, b}$ for $A$, a design matrix for the model, and $b$, a vector of the margins, if we allow cell counts $X_{ijk} \geq -1$. 
\end{proposition}
\begin{comment}
\begin{proof}
We use elements of all indispensable moves for  $3 \times 4\times K$ tables in \cite{Aoki03thelist}.  Then as we discussed in Remark \ref{rm:indis}, we can show the positive component and negative component of each indispensable move can be connected via basic moves.  
\begin{itemize}
    \item For $2 \times 3 \times 3$ indispensable move:
    \[
    {\bf b}_- = {\bf b}_+ + (1,2; 2, 3; 2,3) - (1, 2; 1, 2; 1, 2).
    \]
    \item $2 \times 4 \times 4$ indispensable move:
    \[
    {\bf b}_- = {\bf b}_+ - (1,2; 1, 2; 1,2) - (1, 2; 1, 3; 2, 4) - (1,2;3,4;3,4).    
    \]
    \item $3 \times 3 \times 4$ indispensable move:
    \[
    {\bf b}_- = {\bf b}_+ - (1,2; 1, 2; 2,3) + (2, 3; 2, 3; 2, 3) + (2,3;2,4;1,3).    
    \]
    \item $3 \times 4 \times 4$ move of degree 9  indispensable move:
    \[
    {\bf b}_- = {\bf b}_+ - (1,2; 1, 4; 1,4) - (2, 3; 3, 4; 3, 4) + (1,3;2,4;2,4).    
    \]
    \item $3 \times 3 \times 5$ move of degree 10  indispensable move:
    \[
    {\bf b}_- = {\bf b}_+ - (1,2; 1, 4; 1,4) - (2, 3; 3, 4; 3, 4) + (1,3;2,4;2,4).    
    \]
\end{itemize}
\end{proof}
\end{comment}

In addition, we can show that a Markov chain with the set of all basic moves by allowing all cell counts to be negative one if $I = 2$. 
\begin{proposition}\label{pro1}
    Let $M$ be a set of moves. Suppose all moves in $M$ connect all
tables in ${\mathcal{F}}_{A, b}$ if $b > 0$.  Then $M$ connects all tables in ${\mathcal{F}}_{A, b'}$ if we allow cell counts to be $-1$.
\end{proposition}
\begin{proof}
    We consider the fiber 
\[
\begin{array}{ccc}
\bar{\mathcal{F}}_{A, b}&:= &\{ u\in \Z^n: A u = b, \, u \geq -\mathbb{I}_n\}\\
&=& \{ u + \mathbb{I}_n\in \Z_{\geq 0}^n: A (u +
\mathbb{I}_n)= b\}\\
&= & \{ u\in \Z_{\geq 0}^n: A u = (b - A
\mathbb{I}_n)\}.\\
\end{array}
\]
Since we have $ (b - A
\mathbb{I}_n) > 0$,  $M$ connects all tables in $\bar{\mathcal{F}}_{A, b}$.
\end{proof}

\begin{theorem}\label{th:2JK}
    Suppose we consider $2 \times J \times K$ contingency tables under the no-three-way interaction model.  Then the set of all basic moves connects tables in the fiber ${\mathcal{F}}_{A, b}$ for $A$, a design matrix for the model, and $b$, a vector of the margins, if we allow cell counts $X_{ijk} \geq -1$. 
\end{theorem}
\begin{proof}
    Theorem 3 in \cite{RY} shows that the set of basic moves connects bounded $J \times K$ tables  under the independence model if we have positive margins.  This is equivalent that the set of all basic moves connect tables for $2 \times J \times K$ tables under the no-three-way interaction model with all positive margins.  Then we apply Proposition \ref{pro1} we have the result. 
\end{proof}

In general we do not know that a set of basic moves of $2 \times 2 \times 2$ minors connects all tables in $\mathcal{F}_{A, b}$ under the no-three-way interaction model for any $b \in \mathbb{Z}^{m}_{\geq 0}$ if we allow cell counts to be $-1$.  
Thus, it is still an open problem to prove that moves of $2\times 2\times 2$ minors connects all $I \times J \times K$ tables in the fiber for $I > 3, \, J > 3, \, K>3$ under the no-three-way interaction model if we allow cell counts to be negative one.  

\section{USN Officer and Enlisted Population by Race, Rank, and Gender}\label{sec:comp}

We now examine the performance of the algorithms using an empirical data set from \cite{NavyData}.  The data set consists of the USN officer and enlisted population as of January 20, 2021.  The data is organized into a three way contingency table stratified by race, rank, and gender.  The rank category consists of 19 levels.  The values are displayed in Tables~\ref{tab:maloff},~\ref{tab:malenl},~\ref{tab:femoff}, and~\ref{tab:femenl}. By Theorem \ref{th:2JK}, we know that the set of all basic moves connects tables in the fiber if we allow cell counts to be $-1$.  Thus we apply our MCMC scheme to this contingency table.

\begin{table}[H]
\caption{Male Naval Officer Population by Race and Rank as of January 20,2021}\label{tab:maloff}
\begin{center}
\begin{tabular}{|c|c|c|c|c|c|c|c|c|c|c|}
    \hline
    Race&Adm.&O-6&O-5&O-4&O-3&O-2&O-1&W-4&W-3&W-2\\\hline
    Nat. Am.&1&9&29&92&154&64&66&7&14&15\\ \hline
    Asian&2&96&225&417&832&347&349&12&51&34 \\ \hline
    Af. Am.&6&167&357&551&1,006&329&396&81&148&142\\\hline
    Pac. Isl.&2&107&263&280&413&127&118&14&38&50\\\hline
    Multi-Race&2&36&123&242&853&311&342&13&19&14\\\hline
    White&192&2,452&4,752&6,517&11,635&4,038&4,038&222&413&371\\
    \hline
\end{tabular}
\end{center}
\end{table}

\begin{table}[H]
\caption{Male Enlisted Population by Race and Rank as of January 20,2021}
\begin{center}
\begin{tabular}{|c|c|c|c|c|c|c|c|c|c|}
    \hline
    Race&E-9&E-8&E-7&E-6&E-5&E-4&E-3&E-2&E-1\\\hline
    Nat. Am&44&210&777&1,863&1,142&379&254&110&56\\ \hline
    Asian&129&466&1,257&2,720&3,405&2,229&2,051&693&452 \\ \hline
    Af. Am.&409&994&3,151&7,514&9,963&6,533&5,968&2,820&2,036\\\hline
    Pac. Isl.&93&308&798&1,387&2,593&3,087&1,922&333&30\\\hline
    Multi-Race&48&162&942&4,679&5,004&2,189&1,279&476&348\\\hline
    White&1,809&4,242&11,591&26,435&34,716&25,716&20,871&9,140&6,502\\
    \hline
\end{tabular}
\end{center}
\label{tab:malenl}
\end{table}

\begin{table}[H]
\caption{Female Naval Officer Population by Race and Rank as of January 20,2021}
\begin{center}
\begin{tabular}{|c|c|c|c|c|c|c|c|c|c|c|}
    \hline
    Race&Adm.&O-6&O-5&O-4&O-3&O-2&O-1&W-4&W-3&W-2\\\hline
    Nat. Am.&0&2&11&12&45&27&32&0&0&0\\ \hline
    Asian&1&28&56&136&323&129&126&0&3&4 \\ \hline
    Af. Am.&0&50&103&226&466&143&141&13&29&38\\\hline
    Pac. Isl.&1&14&39&81&215&55&45&0&8&6\\\hline
    Multi-Race&0&5&33&87&327&130&144&1&3&1\\\hline
    White&13&294&677&1,447&2,955&1,089&1,117&12&21&33\\
    \hline
\end{tabular}
\end{center}
\label{tab:femoff}
\end{table}

\begin{table}[H]
\caption{Female Enlisted Population by Race and Rank as of January 20,2021}
\begin{center}
\begin{tabular}{|c|c|c|c|c|c|c|c|c|c|}
    \hline
    Race&E-9&E-8&E-7&E-6&E-5&E-4&E-3&E-2&E-1\\\hline
    Nat. Am&6&24&114&294&314&173&124&53&23\\ \hline
    Asian&14&71&214&545&863&652&677&198&129 \\ \hline
    Af. Am.&88&915&3,151&2,509&4,453&3,039&2,983&1,050&894\\\hline
    Pac. Isl.&13&50&157&297&779&920&714&132&5\\\hline
    Multi-Race&5&24&181&975&1,536&757&544&165&122\\\hline
    White&120&324&1,267&3,536&7,555&6,669&5,914&2,308&1,748\\
    \hline
\end{tabular}
\end{center}
\label{tab:femenl}
\end{table}

Of note, the table consists of a wide range of cell counts.  In general, cell counts decrease as rank increases which is expected as there are fewer billets for higher ranks thus reflecting the USN's ``up or out'' promotion structure.  Further, we note that there are considerably fewer females than males across all ranks and races.  For this reason, several sub-categories were combined to ensure that the chain runs smoothly. 

In this computation, we consider the following hypotheses:
\[
\begin{array}{cc}
H_0:& \lambda_{ijk}^{XYZ} = 0 \mbox{ for all } i = 1, \ldots, 10, \, j= 1, \ldots , 6, \, k = 1, 2.\\
& \\
 H_1:    & \lambda_{ijk}^{XYZ} \not = 0 \mbox{ for all } i = 1, \ldots, 10, \, j= 1, \ldots , 6, \, k = 1, 2.
\end{array}
\]

We estimated maximum likelihood estimate (MLE) via Iterative Proportional Fitting Procedure (IPFP) under the no-three-way interaction model shown in Tables Tables~\ref{tab:mlemaloff},~\ref{tab:mlemalenl},~\ref{tab:mlefemoff}, and~\ref{tab:mlefemenl}.
\begin{table}[H]
\caption{MLE under the no-three-way interaction model for Male Naval Officer Population by Race and Rank as of January 20,2021}
\begin{center}
{\tiny
\begin{tabular}{|c|c|c|c|c|c|c|c|c|c|c|}
    \hline
    Race&Adm.&O-6&O-5&O-4&O-3&O-2&O-1&W-4&W-3&W-2\\\hline
    Nat. Am.&0.92 &   9.52  & 33.91  & 81.65 &  150.63 &  68.19 &   73.55 &  6.55 & 12.88 & 13.39\\ \hline
    Asian&2.76&  107.26&  238.10&  433.88 &  873.60&  356.39 & 356.20  &11.22  &49.66&  33.90 \\ \hline
    Af. Am.&5.17&  168.82  &346.00 & 517.31 &  926.37 & 292.50 & 333.57 & 83.44 &152.62 &147.43\\\hline
    Pac. Isl.&2.74 & 103.43 & 252.48 & 277.95 &  464.99  &133.31 & 119.59 & 13.02 & 42.01&  49.49 \\\hline
    Multi-Race&1.81  & 34.71 & 129.01  &249.48  & 858.83 & 317.35 & 350.35 & 12.96 & 19.97  &13.15\\\hline
    White&191.60 &2443.25 &4749.50& 6538.72 &11618.57 &4048.27 &4075.74& 221.81 &405.85& 368.64\\
    \hline
\end{tabular}}
\end{center}
\label{tab:mlemaloff}
\end{table}

\begin{table}[H]
\caption{MLE under the no-three-way interaction model for Male Enlisted Population by Race and Rank as of January 20,2021}
\begin{center}
{\tiny
\begin{tabular}{|c|c|c|c|c|c|c|c|c|c|}
    \hline
    Race&E-9&E-8&E-7&E-6&E-5&E-4&E-3&E-2&E-1\\\hline
    Nat. Am&45.58&193.81 &  711.45 & 1835.01 & 1152.64 &  424.38  & 283.95 & 127.27 &  60.72\\ \hline
    Asian&130.32&444.54 & 1173.84 & 2776.32 & 3376.57&  2213.34 & 2047.63 & 695.21 & 446.27 \\ \hline
    Af. Am.&421.94&1383.03 & 4307.92  &7582.58 & 9722.60 & 6169.24 & 5568.08 &2554.54 &1887.83\\\hline
    Pac. Isl.&95.85&291.92 &  748.72 & 1413.27&  2619.39 & 3016.65&  1935.93 & 355.91 &  26.34\\\hline
    Multi-Race&47.62&149.82  & 867.92 & 4694.68 & 5005.78 & 2181.91&  1315.54&  483.06 & 348.02\\\hline
    White&1790.69&3918.87 &10706.15 &26296.14 &34946.02 &26127.49 &21193.87 &9356.01& 6654.81\\
    \hline
\end{tabular}}
\end{center}
\label{tab:mlemalenl}
\end{table}

\begin{table}[H]
\caption{MLE under the no-three-way interaction model for Female Naval Officer Population by Race and Rank as of January 20,2021}
\begin{center}
{\tiny 
\begin{tabular}{|c|c|c|c|c|c|c|c|c|c|c|}
    \hline
    Race&Adm.&O-6&O-5&O-4&O-3&O-2&O-1&W-4&W-3&W-2\\\hline
    Nat. Am.&0.08  & 1.48  & 6.09 &  22.35  & 48.37 &  22.81 &  24.45 & 0.45 & 1.12 & 1.61\\ \hline
    Asian&0.24 & 16.74 & 42.90 & 119.12 & 281.40 & 119.61 & 118.80  &0.78  &4.34&  4.10 \\ \hline
    Af. Am.&0.83 & 48.18 &114.00 & 259.69 & 545.63 & 179.50 & 203.43& 10.56 &24.38 &32.57\\\hline
    Pac. Isl.&0.26 & 17.57&  49.52 &  83.05 & 163.01 &  48.69 &  43.41 & 0.98 & 3.99 & 6.51\\\hline
    Multi-Race&0.19 &  6.29 & 26.99 &  79.52 & 321.17 & 123.65  &135.65&  1.04 & 2.03 & 1.85\\\hline
    White&13.40 &302.75& 679.50 &1425.28 &2971.43 &1078.73 &1079.26& 12.19& 28.15 &35.36\\
    \hline
\end{tabular}}
\end{center}
\label{tab:mlefemoff}
\end{table}

\begin{table}[H]
\caption{MLE under the no-three-way interaction model for Female Enlisted Population by Race and Rank as of January 20,2021}
\begin{center}
{\tiny
\begin{tabular}{|c|c|c|c|c|c|c|c|c|c|}
    \hline
    Race&E-9&E-8&E-7&E-6&E-5&E-4&E-3&E-2&E-1\\\hline
    Nat. Am&4.42 & 40.19&179.55&  321.99  &303.36 & 127.62 &  94.05&   35.73  & 18.28\\ \hline
    Asian&12.68 & 92.46&297.16 & 488.68 & 891.43 & 667.66 & 680.37 & 195.79 & 134.73 \\ \hline
    Af. Am.&75.06 & 525.97&1994.08 &2440.41 &4693.40& 3402.76& 3382.92 &1315.46& 1042.17\\\hline
    Pac. Isl.&10.15 & 66.08&206.28 & 270.73 & 752.61 & 990.35&  700.07&  109.09 &   8.66\\\hline
    Multi-Race& 5.38 & 36.18&255.08 & 959.32& 1534.22 & 764.09  &507.46 & 157.94&  121.98\\\hline
    White&138.32 &647.134&2151.85 &3674.87& 7324.98& 6257.51& 5591.13 &2091.99& 1595.19\\
    \hline
\end{tabular}}
\end{center}
\label{tab:mlefemenl}
\end{table} 

In this study, as the first experiment, we only look at ranks which include Admirals, officers (O-1 through O-6) and W-4, W-3, and W-2. The race category consists of six levels and gender is binary.  This results in a $10 \times 6 \times 2$ contingency table.  The $\chi^2$ test statistics of the observed table is $90.23$. 
Therefore an estimate p-value using the $\chi^2$ distribution with $(10-1)\cdot (6-1) \cdot (2-1)$ is $0.00002$.

We took the sample size $N = 10000$ with 25\% burn-in and thinning 25.  With our MCMC scheme with the set of all basic moves by allowing cell counts to be negative one, we have an estimated p-value for this hypothesis test $0.00001$.  If we set the significance level $\alpha = 0.01$, we reject the null hypothesis and we support the alternative.  Figure \ref{fig:hist} shows that an estimated distribution of $\chi^2$ statistics with the fixed margins for this Navy personnel dataset using our MCMC scheme.
\begin{figure}
    \centering
    \includegraphics[width=0.5\textwidth]{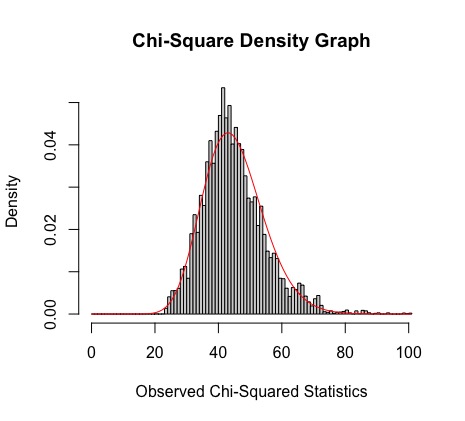}
    \caption{Estimated distribution of $\chi^2$ statistics using our MCMC scheme using the set of basic moves with allowing cell counts to be negative one. The red line is the $\chi^2$ distribution with $(10-1)\cdot (6-5) \cdot (2-1)$ degrees of freedom.}
    \label{fig:hist}
\end{figure}

For the second experiment, we  look at all 19 ranks. The race category consists of six levels and gender is binary.  This results in a $19 \times 6 \times 2$ contingency table.  The $\chi^2$ test statistics of the observed table is $2775.15$. 
Therefore an estimate p-value using the $\chi^2$ distribution with $(19-1)\cdot (6-1) \cdot (2-1)$ is very close to $0$.

We took the sample size $N = 10000$ and thinning 25 with the set of basic moves by allowing cell counts to be $-1$.  For burn-in step, due to its size, we took 250\% burn-in with the set of basic moves by allowing only nonnegative cell counts since this is just a burn-in step and since this chain has much better mixing time.  With our MCMC scheme with the set of all basic moves by allowing cell counts to be negative one, we have an estimated p-value for this hypothesis test $< 0.00001$.  If we set the significance level $\alpha = 0.01$, we reject the null hypothesis and we support the alternative.  Figure \ref{fig:hist2} shows that an estimated distribution of $\chi^2$ statistics with the fixed margins for this Navy personnel dataset using our MCMC scheme.
\begin{figure}
    \centering
    \includegraphics[width=0.5\textwidth]{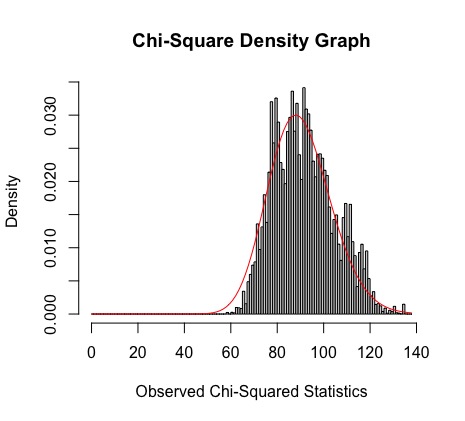}
    \caption{Estimated distribution of $\chi^2$ statistics using our MCMC scheme using the set of basic moves with allowing cell counts to be negative one. The sample size of a Markov chain is $N = 10000$. The red line is the $\chi^2$ distribution with $(19-1)\cdot (6-5) \cdot (2-1)$ degrees of freedom.}
    \label{fig:hist2}
\end{figure}

Since the size of a $19 \times 6 \times 2$ contingency table is a large table, we increased the sample size of a Markov chain to $N = 50000$ while the burn-in sample size is fixed as $B = 2500000$.  Then we have an estimated p-value for this hypothesis test $< 0.000005$.  If we set the significance level $\alpha = 0.01$, we reject the null hypothesis and we support the alternative.  Figure \ref{fig:hist3} shows that an estimated distribution of $\chi^2$ statistics with the fixed margins for this Navy personnel dataset using our MCMC scheme.
\begin{figure}
    \centering
    \includegraphics[width=0.5\textwidth]{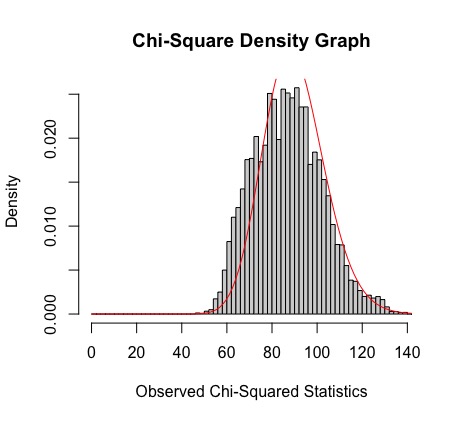}
    \caption{Estimated distribution of $\chi^2$ statistics using our MCMC scheme using the set of basic moves with allowing cell counts to be negative one. The sample size of a Markov chain is $N = 50000$.  The red line is the $\chi^2$ distribution with $(19-1)\cdot (6-5) \cdot (2-1)$ degrees of freedom.}
    \label{fig:hist3}
\end{figure}

\section{Conclusion}

In this paper we conduct MCMC with basic moves by allowing cell counts to be negative one. %Additionally, all of the algorithms allow movement outside the lattice to facilitate movement to other parts of $\mathcal{F}_{A, b}$.  
Experimental results imply that for the case of the three-way $I \times J \times K$ tables under the no-three-way interaction model, our MCMC algorithm provides an estimated distribution of $\chi^2$ statistics closer to the asymptotic distribution, $\chi^2$ distribution with the degrees of freedom $(I-1)(J-1)(K-1)$.  

In addition, it is still an open problem to prove that $I \times J \times K$ tables for $I > 3, \, J > 3, \, K>3$ under the no-three-way interaction model moves of $2\times 2\times 2$ minors connect all tables in the fiber if we allow cell counts to be negative one.  Specifically 
\begin{conjecture}
Consider a three way $I \times J\times J$ contingency table
$X_{ijk}$ for $I > 3, \, J > 3, \, K>3$ under the no-three-way interaction model.
If we allow $X_{ijk} \geq -1$, then all moves of $2\times 2\times 2$ minors connect all tables in $\mathcal{F}_{A, b}$, where $A$ is a design matrix for the model and $b$ is a vector of margins computed from the observed table.
\end{conjecture}
\section*{Acknowledgement}

RY and DB are partially funded by NSF Division of Mathematical Sciences: Statistics Program DMS 1916037. 

\bibliography{refs}
\bibliographystyle{amsalpha}
\end{document}